\documentclass[11pt,letterpaper]{article}
\usepackage[margin=1in]{geometry}
\usepackage{amsmath,amssymb,amsfonts,setspace,url}
\usepackage{latexsym}
\usepackage{amsthm}
\usepackage{subfig}
\usepackage{float}
\usepackage{algorithmic}
\usepackage[ruled]{algorithm}
\usepackage{citesort}
\usepackage{bigstrut,array,multirow}
\usepackage{here}

\theoremstyle{plain}
\newtheorem{theorem}{Theorem}
\newtheorem{lemma}{Lemma}

\newtheorem{definition}{Definition}
\newtheorem{fact}{Fact}

\newtheorem{proposition}{Proposition}

\theoremstyle{definition}

\floatstyle{ruled}
\newfloat{algorithmfig}{thp}{lof}
\floatname{algorithmfig}{Algorithm}

\newcommand{\Omit}[1]{}

\newcounter{Codeline}

\newcommand{\vecb}[1]{\boldsymbol{#1}}

\newcommand{\trih}[1]{\ensuremath{{T}_{#1}(G)}}
\newcommand{\multri}[1]{\ensuremath{\#(#1)}}
\newcommand{\DD}{\ensuremath{\vecb{T}_{w(\vecb{T})}}}

\newcommand{\X}{X}
\newcommand{\mydelta}{\Delta(\X)}
\newcommand{\myS}[1]{\mathcal{S}_U^{\X}(#1)}
\providecommand{\m}{r} %new macro
\newcommand{\myT}[1]{\mathcal{V}_{U,\m}^{\X}(#1)}
\newcommand{\E}{\mathcal{E}}
\newcommand{\Pp}{\mathcal{P}}
\newcommand{\ceil}[1]{\left\lceil #1 \right\rceil}

\providecommand{\Aa}{\mathcal{A}}

\providecommand{\Hh}{\mathcal{H}}
\providecommand{\Ii}{\mathcal{I}}

\newcommand{\Hline}[1]{\noalign{\hrule height #1}} 
\begin{document}

\title{Triangle Finding and Listing in CONGEST Networks}
\author{Taisuke Izumi\footnote{Postal address: Nagoya Institute of Technology, Gokiso-cho, Showa-ku, Nagoya, Aichi, 466-8555, Japan. Tel: (+81)-52-735-5567.}\\
Graduate School of Engineering\\
Nagoya Institute of Technology\\
\url{t-izumi@nitech.ac.jp}
\and
Fran{\c c}ois Le Gall\footnote{Postal address: Department of Communications and Computer Engineering,
Graduate School of Informatics, Kyoto University, Yoshida-Honmachi, Sakyo-ku, Kyoto 606-8501, Japan. Tel: (+81)-75-753-5382.}
\\Graduate School of Informatics\\
Kyoto University\\
\url{legall@i.kyoto-u.ac.jp}
}
\date{}

\maketitle
\thispagestyle{empty}
\setcounter{page}{0}
\begin{abstract}
Triangle-free graphs play a central role in graph theory, and triangle detection (or triangle finding) as well as triangle enumeration (triangle listing) play central roles in the field of graph algorithms. In distributed computing, algorithms with sublinear round complexity for triangle finding and listing have recently been developed in the powerful CONGEST clique model, where communication is allowed between any two nodes of the network. In this paper we present the first algorithms with sublinear complexity for triangle finding and triangle listing in the standard CONGEST model, where the communication topology is the same as the topology of the network. More precisely, we give randomized algorithms for triangle finding and listing with round complexity $O(n^{2/3}(\log n)^{2/3})$ and $O(n^{3/4}\log n)$, respectively, where $n$ denotes the number of nodes of the network. We also show a lower bound $\Omega(n^{1/3}/\log n)$ on the round complexity of triangle listing, which also holds for the CONGEST clique model. 
\end{abstract}
\newpage

%==============================================
%=============================================
\section{Introduction}
%==============================================
%=============================================
\paragraph{Background.}
%==============================================
The most standard way of modeling distributed networks is to use graphs, where vertices and edges respectively correspond to computing entities (e.g., nodes or processes) and communication channels. The graphical structure of networks naturally motivates, both practically and theoretically, the study of graph algorithms, i.e., in-network solving of combinatorial graph problems for the own topology of the network.  In the last twenty years, research on distributed complexity theory has successfully investigated many fundamental graph problems such as independent set~\cite{ABI86,BEPS16}, dominating set~\cite{FR05,KMW16}, coloring~\cite{Linial92,BE13}, minimum spanning trees~\cite{GHS83,Elkin04,Elkin06,GKP93,PR99,LPPP05} shortest paths~\cite{LP13,LP15,Nanongkai14,HW12}, max-flow and min-cut~\cite{GK13,NS14,GKKLP15}. 
%Independent set, dominating set, coloring, minimum spanning trees, 
%shortest paths, max-flow and min-cut, and so on. 
The seminal textbook by Peleg~\cite{Peleg00} provides a good overview of those achievements.

A high-level perspective on distributed graph algorithms classifies problems into two categories according to their locality characteristics. Problems like minimum spanning trees, shortest paths or max-flow/min-cut mentioned above are classified as global problems because some node must (indirectly) communicate with a node $D$-hop away, where $D$ denotes the diameter of the network, which requires $\Omega(D)$ rounds. Since any problem can be solved within $O(D)$ rounds by a naive centralized approach if the communication bandwidth of each channel is unlimited, technical challenges on global problems appear only in the model of limited-bandwidth channels, the so-called CONGEST communication model where each channel has only $O(\log n)$-bit bandwidth. On the other hand, many local problems, like independent set, dominating set or coloring allow distributed (approximate) solutions within $o(D)$ rounds. Besides the limited bandwidth, a significant difficulty when designing distributed algorithms for local problems often consists in breaking the symmetry between nearby nodes.

\paragraph{Triangle finding.}
Triangle finding is an important counterpart of the perspective explained above. Since the distance-two information for a node $i$ is sufficient to find all triangles around~$i$, it is a purely local task. Aggregating at each node the set of its 2-hop neighbors, however, takes $\Theta(d_{\mathit{max}})$ rounds in the CONGEST model, where $d_{\mathit{max}}$ denotes the maximum degree of the nodes. In the case of dense graphs, such an approach gives linear round complexity. In this sense, despite being a local problem, the underlying difficulty of triangle finding is similar to the difficulty encountered when working with global problems in the CONGEST model, that is, the lack of communication bandwidth. These particularities make triangle finding a difficult problem to handle with current techniques. As summarized in Table~\ref{table:comparison}, recently some non-trivial upper bounds have been obtained by Dolev et al.~\cite{DLP12} and then improved by Censor-Hillel et al.~\cite{CKKLPS15}, but these upper bounds hold only in the much stronger CONGEST clique model (where at each round messages can even be sent between non-adjacent nodes) for which bandwidth management is significantly easier. No nontrivial upper bound is known in the standard CONGEST model.
As far as lower bounds are concerned, the locality of triangle finding rules out many standard approaches used to obtain lower bounds on the complexity of global problems  (e.g., the approaches from \cite{SHKKANPPW12,FHW12}). The only known lower bounds have been obtained recently by Drucker et al.~\cite{DKO14}, but hold only for the much weaker broadcast CONGEST model, where at each round the nodes can broadcast only a single common message to all other nodes, under a (reasonable) conjecture in communication complexity theory. No nontrivial lower bound is known for triangle finding in the standard CONGEST model.

Another, more practical, motivation for considering triangle finding
%Besides these theoretical considerations, a practical motivation for considering triangle finding 
%(and its generalization triangle listing described later) 
is that for several graph problems faster algorithms are known over triangle-free graphs (we refer to \cite{Hirvonen+14,Pettie+15} for some examples in the distributed setting). The ability to efficiently check if the network is triangle-free, and more generally detect which part of the network is triangle-free, is essential when considering such algorithms in practice. 

\setlength{\extrarowheight}{1.5pt}
\begin{table*}[tb]
 \begin{center}
  \begin{tabular}{cccc}
  \Hline{1.5pt}
   Paper & Time Bound & Problem & Model \\ \Hline{1.5pt}
   \multirow{2}{*}{Dolev et al.~\cite{DLP12}} & $O(n^{1/3}(\log n)^{2/3})$& \multirow{2}{*}{Listing}&\multirow{2}{*}{CONGEST clique}\bigstrut\\
   &$O(d_{\mathit{max}}^{3}/n)$ && \\ \hline
   Censor-Hillel et al.~\cite{CKKLPS15} & 
   %$O(n^{1-2/\omega}) = 
   $O(n^{0.1572})$ & Finding 
   & CONGEST clique \bigstrut\\ \hline
%   Le Gall\cite{LeGall16} & $\tilde{O}(n^{1-2/\omega}) = \tilde{O}(n^{0.1572}) $ &
%   Counting & CONGEST clique \\ \hline 
   This paper (Theorem \ref{th:UB-finding}) & $O(n^{2/3}(\log n)^{2/3})$ & Finding & CONGEST \bigstrut\\ \hline 
   This paper (Theorem \ref{th:UB-listing})& $O(n^{3/4}\log n)$ & Listing & CONGEST \bigstrut\\ \Hline{1.5pt}

   Drucker et al.~\cite{DKO14} & 
   $\Omega\big(\frac{n}{e^{\sqrt{\log n}}\log n}\big)$ (conditional)
   & Finding & CONGEST broadcast \bigstrut\\ \hline
   Pandurangan et al.~\cite{Pandurangan+16}& $\Omega(\frac{n^{1/3}}{\log^3 n})$ & Listing & CONGEST clique \bigstrut\\ \hline
   This paper (Theorem \ref{thm:lowerBound})& $\Omega(\frac{n^{1/3}}{\log n})$ & Listing & CONGEST clique \bigstrut\\ \Hline{1.5pt}
  \end{tabular}
  
   \caption{Prior results on the round complexity of distributed triangle finding and listing and our new results, where $n$ denotes the number of nodes in the network.}
  \label{table:comparison}  
 \end{center}
\end{table*}

%==============================================
\paragraph{Technical contribution.}
%==============================================
In this paper, in addition to triangle finding we also consider the triangle listing problem, which requires that each triangle of the network should be output by at least one node (we refer to Section \ref{sec:prelim} for the formal definition). The triangle listing problem can be seen as a special case of motif finding, which is a popular problem in the context of network data analysis. It is obvious that triangle finding is not harder than triangle listing, since finding reduces to listing.
%A third variant, the triangle finding problem, requires to count the number of triangles. 
In the literature mentioned above, triangle finding and triangle listing are not explicitly distinguished, but the algorithms by Dolev et al.~\cite{DLP12} are actually for the listing version, while the algorithms by Censor-Hillel et al.~\cite{CKKLPS15} are based on an algebraic approach and only support finding. The lower bound by Drucker et al.~\cite{DKO14} applies to triangle finding. 

Our results are summarized in Table~\ref{table:comparison}. Our main contributions are a $O(n^{2/3}(\log n)^{2/3})$-round algorithm for triangle finding and a $O(n^{3/4}\log n)$-round algorithm for triangle listing, both in the CONGEST model. These two algorithms are the first algorithms with sublinear round complexity for triangle finding or listing applicable to the standard CONGEST model. The existence of sublinear-time algorithms shows, from an algorithmic perspective as well, that these two problems are indeed easier than global problems like diameter computation, which has a nearly linear-time lower bound in the CONGEST model~\cite{FHW12}. Interestingly, this result contrasts with the known relations between triangle finding and diameter computation in the centralized setting, in which the best known algorithms for both problems rely on matrix multiplication and have the same complexity. 
%triangle finding and listing the locality of the problem can indeed be exploited. 

%By a simple information-theoretic argument, 
We also show that there exist networks of $n$ nodes for which any triangle listing algorithm requires $\Omega(n^{1/3}/\log n)$ rounds in the CONGEST model. Actually we show a stronger result: any triangle listing algorithm requires $\Omega(n^{1/3}/\log n)$ rounds even in the CONGEST clique model. This lower bound slightly improves the $\Omega(n^{1/3}/\log^3 n)$-round lower bound for triangle listing in the CONGEST clique model obtained recently by Pandurangan et al.~\cite{Pandurangan+16}. Note that since the algorithm by Dolev et al.~\cite{DLP12} has round complexity $O(n^{1/3}(\log n)^{2/3})$, these lower bounds are tight up to possible polylogarithmic factors in the CONGEST clique model. These lower bounds also have the following interesting consequence. They shows that triangle listing is strictly harder than triangle finding in the CONGEST clique model since, as already mentioned, $(n^{0.1572})$-round algorithms for triangle finding have been shown by Censor-Hillel et al.~\cite{CKKLPS15}. Actually, since the algorithms in~\cite{CKKLPS15} also count the number of triangles, these lower bounds thus even show that triangle listing is harder than counting in this model. 

Our lower bound is shown by a simple information-theoretic argument on random graphs. By similar arguments (see Proposition~\ref{cor} in Section~\ref{sec:LB}), we can also show a $\Omega(n/\log n)$-round lower bound for local listing algorithms, where each node~$i$ is required to output all the triangles containing~$i$. This corollary implies that any triangle listing algorithm with sublinear complexity inherently requires some counter-intuitive mechanism that enables a triangle~$t$ to be output by a node not contained in $t$. This is precisely how our $O(n^{3/4}\log n)$-round algorithm works.

%Table~\ref{table:comparison} summarizes our results and compare them with the previous results related to distributed triangle finding and listing.
%==============================================
\paragraph{Other related works.}
%==============================================
Distributed algorithms for triangle-freeness have recently also been considered in the setting of property testing \cite{Censor-Hillel+DISC16} (see also \cite{Fraigniaud+DISC16}). Note that in property testing the goal is to decide whether either the network is triangle-free or it contains a large number triangles. Since there is no need to consider the case where the network may have only a small number of triangles, this property testing version is significantly easier than the problems studied in the present paper.   

Triangle finding is a fundamental problem in the field of centralized algorithms as well. 
%\cite{Bansal+ToC12,Czumaj+SICOMP09,Itai+SICOMP78,PatrascuSTOC10,Vassilevska+SICOMP13,Williams+FOCS10,WilliamsSTOC14}. 
It has been known for a long time that this problem is not harder than Boolean matrix multiplication~\cite{Itai+SICOMP78}.
%, which implies that triangle finding in a graph of $n$ vertices can be solved in $O(n^{\omega+\varepsilon})$ time for any constant $\varepsilon>0$, where~$\omega<2.38$ represents the exponent of square matrix multiplication. This is still the best known upper bound on the centralized time complexity of triangle finding. 
A few years ago Vassilevska Williams and Williams showed a converse reduction~\cite{Williams+FOCS10}: they proved that a subcubic-time algorithm for triangle finding can be used, in a combinatorial way, to construct a subcubic-time algorithm for Boolean matrix multiplication. This result, combined with recent developments \cite{Bansal+ToC12,Czumaj+SICOMP09,PatrascuSTOC10,Vassilevska+SICOMP13}, have put triangle finding as a central problem in the recent theory of fine-grained complexity.

%==============================================
%==============================================
\section{Preliminaries}\label{sec:prelim}
%==============================================
In this paper we consider undirected graphs. We will use $G = (V, E)$ to denote the graph considered, and write $n=|V|$ and $m=|E|$. In this section we present some of our notations, give details of our model and present some facts about hashing functions that will be used in Section \ref{sec:UB}.\vspace{-2mm}

\paragraph{Graph-theoretic notations.}
 For any vertex $i\in V$, we denote $\mathcal{N}(i)$ the set of neighbors of $i$. For any finite set $X$, we will use the notation $\E(X)$ to represent the set of unordered pairs of elements in~$X$, and use $\mathcal{T}(X)$ to represent the set of unordered triples of elements in~$X$. We will write $\E=\E(V)$ and $\mathcal{T}=\mathcal{T}(V)$.
%, and use the notation $M =  |\mathcal{E}|$ and $N = |\mathcal{T}|$. 
Given a pair of vertices $\{j,k\}\in\E$, we write
\[
\multri{\{j,k\}}=
\Big|\left\{
l\in V\:|\:
\{j,l\}\in E \textrm{ and } \{k,l\}\in E
\right\}\Big|.
\]
Given $t = \{j,k,l\} \in \mathcal{T}$, and $e \in \mathcal{E}$, we will write
$e \in t$ when $e = \{j, k\}$, $e = \{j, l\}$, or $e = \{k, l\}$.  For any $R \subseteq \mathcal{T}$, we denote by $\Pp(R)\subseteq E$ the set of edges $e\in E$ such that $e\in t$  for some triple $t \in R$.

We define a triangle $t = \{j, k, l\}\in \mathcal{T}$ to be an unordered triple of vertices where any pair corresponds to an edge in $E$. Note that if $e$ is an edge of the graph, then $\multri{e}$ represents the number of triangles containing $e$. We write $T(G)$ the set of all triangles contained in graph $G$. 

%Given an output pattern $R = (R_0, R_1, \cdots, R_{n-1})$, we define $w(R)$ to be the node $i$ maximizing the cardinality of $R_i$. We also define $E(R) = \{e \in E \:|\: \exists t \in R: e \in t \}$ for $R \subseteq \mathcal{T}$. 
\vspace{-2mm}

%==============================================
\paragraph{Communication Model.} 
%==============================================
In the paper we mainly consider the CONGEST communication model. The graph 
$G = (V, E)$ represents the topology of the network, executions proceed with round-based synchrony and each node can transfer one $O(\log n)$-bit message to each adjacent node per round. All links and nodes (corresponding to the edges and vertices of $G$, respectively) are reliable and suffers no faults. Each node has a distinct identifier from a domain $\Ii$. For simplicity we will assume 
$\Ii=V = [0, n-1]$, but this assumption is not essential and easy to remove as long as $|\Ii| = \mathrm{poly}(n)$. It is also assumed that each node can access infinite sequence of local random bits, that is, the algorithm can be randomized. Initially, each node knows nothing about the topology of the network except the set of edges incident to itself and the value of $n$. All our upper bounds given in Section~\ref{sec:UB} hold for the CONGEST model.

The CONGEST clique model is a powerful variant of the CONGEST model. It allows an algorithm to transfer a $O(\log n)$-bit message per round between any two nodes not necessarily adjacent in~$G$ at each round. This means that in this model the communication topology is the complete graph on the $n$ nodes, and the graph~$G$ only takes the role of input instances to the algorithm. Except for the communication topology, all other features are common with the CONGEST model described in the above paragraph. The CONGEST clique model will be considered only in Section~\ref{sec:LB} when proving our lower bound (note that a lower bound for this model immediately holds for the CONGEST model as well).

%In the following argument, all upper bound results are presented for 
%the standard CONGEST model. On the other hand, the lower bounds applies for
%the CONGEST clique model (and thus also applies to the standard model).
\vspace{-2mm}

%==============================================
\paragraph{Triangle finding and listing in the CONGEST model.}
%==============================================
In the triangle finding problem, at least one node must output a triangle if $T(G)$ is not empty, or all the nodes must output ``not found'' otherwise. In the triangle listing problem, the goal is that for each triangle in $T(G)$ at least one node in $V$ outputs it. 

Formally, we describe the output of an algorithm for triangle finding or listing by an $n$-tuple $T = (T_0, T_1, \dots, T_{n-1})$, where  $T_i \subseteq T(G)$ is the output by node $i$. The algorithm can be deterministic or randomized. In the latter case, note that the condition on the $T_i$'s implies that any triple output at a node should corresponds to a triangle of $G$, and thus the algorithm should be one-sided error. Note also that we do not require that the $T_i$'s are mutually disjoint. We often use notation $T$ as the union of~$T_i$  over all nodes if there is no ambiguity. We say that the algorithm solves the triangle finding problem if its output $T$ satisfies $T\cap T(G)\neq \emptyset$. We say that it solves the triangle listing problem if $T=T(G)$.  \vspace{-2mm}

%==============================================
\paragraph{Hash functions.}
%==============================================
%In this subsections we recall the definition of $k$-wise independent hash functions and give a useful lemma.

Let $\mathcal{X}$  and $\mathcal{Y}$ be two finite sets. 
For any integer $s\ge 1$, a family of hash functions $\Hh = \{h_1, h_2, \dots, h_p\}$, where each $h_i$ is a function from $\mathcal{X}$ to $\mathcal{Y}$, is called \emph{s-wise independent} if for any distinct
$x_1, x_2, \dots, x_s \in \mathcal{X}$ and any $y_1, y_2, \dots y_s 
\in \mathcal{Y}$, a function $h$ sampled from
$\Hh$ uniformly at random satisfies $\Pr[\bigwedge_{1 \leq i \leq s} h(x_i) 
= y_i] = 1 / |\mathcal{Y}|^s$. 
We will use the following lemma, which follows almost immediately from the definition of $s$-wise independence.
%whose proof can be found in the appendix.

%The improved algorithm based on 4-wise 
%independent hash functions is obtained by
%replacing step 1 of Algorithm~\ref{fig:algorithmheavylisting} by the 
%following steps 1a and 1b.
%We show that this modified algorithm is the one attaining Proposition \ref{prop:tri-heavy}. Before the proof, we present an elementary analysis 
%on $k$-wise hash functions (the proof is deferred to the appendix).

%\begin{lemma}\label{l:hashfunction}
%Let $h : \Ii \to [0, n^{\varepsilon/2} - 1]$ be a 6-wise independent hash function,
%and $H(U) = \{x \in U | h(x) = 0\}$ for any subset $U, W \subseteq \Ii$. Then
%for any $x, y \in U$ and $U \subseteq \Ii$, we have 
%\[
%\Pr\left[x,y \in H(U) \bigwedge
%|H(U)| \leq \frac{4|U|}{n^{\varepsilon/2}}\right] \geq \frac{15}{16n^{\varepsilon}}.
%\]
%\end{lemma}
\begin{lemma}\label{l:hashfunction}
Let $\Hh$ be any $3$-wise independent family of hash functions from $\mathcal{X}$ to $\mathcal{Y}$, and $h$ be a function sampled from
$\Hh$ uniformly at random. Let $H(y) = \{x'' \in \mathcal{X}\:|\:h(x'')=y\}$.
For any triple $(x,x',y)\in \mathcal{X}\times \mathcal{X}\times \mathcal{Y}$ 
we have 
\[
\Pr\left[h(x)=h(x')=y \bigwedge 
\Big|H(y)\Big| \leq 4\left(2 + \frac{|\mathcal{X}| - 2}{|\mathcal{Y}|}\right)\right] \geq \frac{3}{4|\mathcal{Y}|^2}.
\]
\end{lemma}

\begin{proof}
Let us write $Z=\{x''\in \mathcal{X}\:|\:h(x'')=y\}$. For any $a\in \mathcal{X}$, let $Z_{a}$ be the indicator random variable corresponding to the event $a \in Z$. By the 3-wise independence property of $h$, we have
\begin{align*}
\lefteqn{E\Big[|Z|\:|\:Z_x=Z_{x'}=1\Big]} \hspace{10mm} \\
&= 2+\sum_{a \in \mathcal{X} \setminus \{x, {x'}\}} E[Z_a\:|\:Z_x=Z_{x'}=1] \\
&= 2+\frac{|\mathcal{X}| - 2}{|\mathcal{Y}|}.
%\le \frac{3 |U|}{|\mathcal{Y}|}.
\end{align*}
By Markov's 
inequality, we have
\begin{align*}
\Pr
\left[
|Z| > 4\left(2+\frac{|\mathcal{X}| - 2}{|\mathcal{Y}|}\right) \ \middle| \  Z_{x}=Z_{{x'}}=1 
\right]
&\leq \frac{1}{4}.
\end{align*}
Since $\Pr[Z_x = Z_{x'} = 1] = 1/|\mathcal{Y}|^2$ holds, we obtain
\begin{align*}
\lefteqn{\Pr\left[Z_x = Z_{x'} = 1 \bigwedge |Z| \le 4\left(2+\frac{|\mathcal{X}| - 2}{|\mathcal{Y}|}\right)\right]} \hspace{10mm} \\
&= \Pr\left[|Z| \le 4\left(2+\frac{|\mathcal{X}| - 2}{|\mathcal{Y}|}\right) \middle| Z_x = Z_{x'} = 1 \right] \\
& \hspace{40mm} \times
\Pr[Z_x = Z_{x'} = 1] \\
&\ge 
%\left(1 - \frac{1}{4}\right) \times \frac{1}{|\mathcal{Y}|^2}=
\frac{3}{4} \times \frac{1}{|\mathcal{Y}|^2},
\end{align*}
as claimed.
\end{proof}

Note that using the standard construction by Wegman and Carter \cite{WC81} we can encode $k$-wise independent hash functions using only $O(k\log |\mathcal{Y}|)$ bits.

%==============================================
%==============================================
\section{Upper Bounds}\label{sec:UB}
%==============================================
%==============================================
In this section we prove the following two theorems.
\begin{theorem}\label{th:UB-finding}
Let $\delta>0$ be any constant.
In the CONGEST model there exists a randomized algorithm for triangle finding with round complexity $O(n^{2/3}(\log n)^{2/3})$, where $n$ denotes the size of the network, and success probability at least $1-\delta$.
\end{theorem}
\begin{theorem}\label{th:UB-listing}
In the CONGEST model there exists a randomized algorithm for triangle listing with round complexity $O(n^{3/4}\log n)$ and success probability at least $1-1/n$, where $n$ denotes the size of the network.
\end{theorem}
Theorems \ref{th:UB-finding} and \ref{th:UB-listing} follow from three algorithms described in Propositions \ref{prop:tri-heavy-finding}--\ref{prop:tri-light} below, which rely on the following concept of heavy triangles. Let $\varepsilon$ be any real number such that $0\le \varepsilon\le 1$.  We say that a triangle $t \in T(G)$ is \emph{$\varepsilon$-heavy} if there exists an edge $e \in t$ such that 
$\multri{e} \geq n^{\varepsilon}$. We define $\trih{\varepsilon} \subseteq T(G)$ as 
the set of all $\varepsilon$-heavy triangles in $T(G)$. The first proposition is almost trivial, and
shows how to find an $\varepsilon$-heavy triangle by a straightforward sampling strategy:
\begin{proposition}\label{prop:tri-heavy-finding}
%Let $\varepsilon$ be any constant such that $0\le\varepsilon\le 1$. 
There exists a $O(n^{1-\varepsilon})$-round randomized algorithm $\Aa_1$  returning a set $T\subseteq T(G)$ such that, if $\trih{\varepsilon}\neq \emptyset$, then $T\cap\trih{\varepsilon}\ \neq \emptyset$ with probability $\Omega(1)$.
\end{proposition}
\begin{proof}
Algorithm $\Aa_1$ is as follows: First, each node $j\in V$ constructs a random set $S_j\subseteq\mathcal{N}(j)$ by including in $S_j$ each element of $\mathcal{N}(j)$ with probability $n^{-\varepsilon}$.  If $|S_j|> 4n^{1-\varepsilon}$, node $j$ then does not send anything to its neighbors. Otherwise (i.e., if $|S_j|\le 4n^{1-\varepsilon}$), node $j$ then sends $S_j$ to each neighbor $k\in\mathcal{N}(j)$, and each such neighbor~$k$ checks if there is a triangle containing $j$, $k$ and a third node in $S_j$, which can be done by computing locally the set $\mathcal{N}(k)\cap S_j$. 

Observe that if $\trih{\varepsilon}\neq \emptyset$ then at least one $\varepsilon$-heavy triangle will be found with constant probability. Indeed, for each edge $\{j,k\}\in E$ such that $\multri{\{j,k\}}\ge n^\varepsilon$, with constant probability at least one node adjacent to both $j$ and~$k$ is included in $S_j$, and $S_j$ is small enough.
\end{proof}
The second proposition is about finding each $\varepsilon$-heavy triangle: 
\begin{proposition}\label{prop:tri-heavy}
%Let $\varepsilon$ be any constant such that $0\le\varepsilon\le 1$. 
There exists a $O(n^{1-\varepsilon/2})$-round randomized algorithm $\Aa_2$ returning a set $T\subseteq T(G)$ such that, for any triangle $t\in \trih{\varepsilon}$, this set $T$ contains $t$ with probability $\Omega(1)$.
\end{proposition}
The third proposition is about finding each triangle that is not  $\varepsilon$-heavy:
\begin{proposition}\label{prop:tri-light}
%Let $\varepsilon$ be any constant such that $0\le\varepsilon\le 1$. 
There exists a $O(n^{1-\varepsilon}+n^{(1+\varepsilon)/2}\log n)$-round randomized algorithm $\Aa_3$ returning a set $T\subseteq T(G)$ such that, for any triangle $t\in T(G)\setminus\trih{\varepsilon}$, this set~$T$ contains $t$ with probability $\Omega(1)$.
\end{proposition}
 
Proposition \ref{prop:tri-heavy} and especially Proposition \ref{prop:tri-light} are the main technical contributions of this section. Their proofs are given in Subsections \ref{sub:heavy} and \ref{sub:light}, respectively.

We now explain how Theorems \ref{th:UB-finding} and \ref{th:UB-listing} immediately follow from these three propositions.
\begin{proof}[Proof of Theorem \ref{th:UB-finding}]
The triangle finding algorithm simply applies Algorithm $\Aa_1$ and then Algorithm $\Aa_3$. 
From Propositions \ref{prop:tri-heavy-finding} and \ref{prop:tri-light} we know that a triangle will thus be found with at least constant probability (if $G$ contains a triangle). For any constant $\delta$ the success probability can be amplified to $1-\delta$ by repeating this process $c$ times for a large enough constant~$c$.

The round complexity of this algorithm is 
$
O(n^{1-\varepsilon} + n^{1-\varepsilon} + \allowbreak n^{(1+\varepsilon)/2}\log n).
$
Choosing $\varepsilon$ such that $n^\varepsilon =\frac{n^{1/3}}{(\log n)^{2/3}}$ gives the claimed upper bound.
\end{proof}

\begin{proof}[Proof of Theorem \ref{th:UB-listing}]
The triangle listing algorithm repeats $\ceil{c \log n}$ times the following process (here $c$ is a large constant): apply Algorithm $\Aa_2$ and then Algorithm $\Aa_3$. 

We now show that this algorithm lists all the triangles of $G$. Let $t$ be a triangle in $T(G)$. From Propositions \ref{prop:tri-heavy} and \ref{prop:tri-light} we know that $t$ is found with at least constant probability at each step. If~$c$ is large enough, repeating this process $\ceil{c\log n}$ times guarantees that $t$ is found with probability at least $1-1/n^4$. From the union bound, we can conclude that all triangles are found with probability at least $1-1/n$.

The round complexity of this algorithm is 
$
O(((n^{1-\varepsilon/2} + n^{1-\varepsilon} \allowbreak + n^{(1+\varepsilon)/2}\log n)\log n).
$
Choosing~$\varepsilon$ such that $n^\varepsilon =\frac{n^{1/2}}{(\log n)^{2}}$ gives the claimed upper bound.
%\qed
\end{proof}

%==============================================
\subsection{Listing all $\varepsilon$-heavy triangles: Proof of Proposition \ref{prop:tri-heavy}}\label{sub:heavy}
%==============================================
The goal of this subsection is to prove Proposition \ref{prop:tri-heavy}. 

The main idea to find $\varepsilon$-heavy triangles is that each node $j$ will decide which information about $\mathcal{N}(j)$ it will send to each neighbor $a$ according to a hash function $h_a$ generated by $a$. More precisely, the function $h_a$ will be taken (by node $a$, and then distributed to all its neighbors) from a 3-wise independent family $\Hh$ of hash functions from~$V$ to $\{0,1,\ldots,\lfloor n^{\varepsilon/2}\rfloor-1\}$, and $j$ will send an edge $\{j,l\}$ to $a$ if $h_a(l)=0$. The complete algorithm, which is Algorithm $\Aa_2$ claimed in Proposition~\ref{prop:tri-heavy}, is described in Figure \ref{fig:algorithmheavylisting} and analyzed below.

\begin{figure*}[ht]
\begin{center}
\fbox{
\begin{minipage}{15 cm} 
\begin{itemize}
\item[1.]
Each node $i\in V$ chooses a function $h_i\colon V\to\{0,1,\ldots,\lfloor n^{\varepsilon/2}\rfloor-1\}$ uniformly at random from a 3-wise independent family  $\Hh$ of hash functions, and sends $h_i$ to all the nodes
in $\mathcal{N}(i)$. 
\item[2.]
Each node $j\in V$ locally computes, for each node $a\in \mathcal{N}(j)$, the set of edges
$E^j_a = \{\{j,l\}\in E \:|\: h_a(l) = 0\}$ and then sends $E^j_a$ to $a$ if $|E^j_a| \leq 8+\frac{4n}{\lfloor n^{\varepsilon/2}\rfloor}$.
\item[3.]For any node $i\in V$, let $F_i$ be the set of edges received by $i$ at the first step. Node $i$
outputs all triples $\{j, k, l\}$ satisfying $\{j,k\}, \{j, l\}, \{k, l\} \in F_i$.
\end{itemize}
\end{minipage}
}
\end{center}
\caption{Algorithm $\Aa_2$ finding any $\varepsilon$-heavy triangle in $O(n^{1 - \varepsilon/2})$ rounds.}
\label{fig:algorithmheavylisting}
\end{figure*}

We first prove the correctness of Algorithm $\Aa_2$.
Let $t = \{j, k, l\}$ be any $\varepsilon$-heavy triangle of~$G$. Let $\{j, k\}$ be the edge shared by at least $n^{\varepsilon}$ triangles and $A$ be the set of nodes such that $\{j, k, a\}$ forms a triangle for each $a \in A$. Consider the function $h_a\in\mathcal{H}$ chosen by node $a$. Applying Lemma~\ref{l:hashfunction} with $|\mathcal{X}|=n$ and $|\mathcal{Y}|=\lfloor n^{\varepsilon/2}\rfloor$,
we obtain the inequality
\[
\Pr\left[h_a(k)=h_a(l)=0  \bigwedge |E_a^j\cup E_a^k|\leq 8+\frac{4n}{\lfloor n^{\varepsilon/2}\rfloor}\right] \geq
%\Pr\left[k,l \in H(U) \bigwedge |H(U)| \leq \frac{2|U|}{n^{\varepsilon/2}}\right]
%\geq 
\frac{3}{4n^{\varepsilon}}.
\]
This bound implies that with probability at least $\frac{3}{4n^{\varepsilon}}$ node $a$ receives $\{j,k\}$ from~$j$, $\{j,l\}$ from~$j$ and $\{k,l\}$ from~$k$ at Step 2, and thus 
$T_a$ (i.e., the output of node $a$) includes $t$ with the same probability.
Since the events $t \not\in T_a$ are mutually independent 
over all $a \in A$, we can bound the probability that no node in $A$ finds 
$t$ as follows:
\[
\Pr\left[\bigcap_{a \in A} t \not\in T_a\right] 
\leq \left(1 - \frac{3}{4n^{\varepsilon}}\right)^{|A|} 
\leq \left(1 - \frac{3}{4n^{\varepsilon}}\right)^{n^{\varepsilon}}
\leq e^{-\Omega(1)}.
\]

The round complexity of Algorithm $\Aa_2$ is $O(n^{1-\varepsilon/2})$ since the hash function $h_i$ sent at Step 1 can be encoded using $O(\log n)$ bits, as explained in Section \ref{sec:prelim}.

This concludes the analysis of Algorithm $\Aa_2$ and the proof of Proposition~\ref{prop:tri-heavy}.
%\end{proof}
%==============================================
\subsection{Finding triangles that are not heavy: Proof of Proposition \ref{prop:tri-light}}\label{sub:light}
%==============================================
The goal of this subsection is to prove Proposition \ref{prop:tri-light}. We first give a brief description of the main ideas underlying our algorithm before giving the full description of Algorithm $\Aa_3$.

%=================
\paragraph{Key definition and outline of our approach.}
%We start with a key definition. 
For any set $\X\subseteq V$, define the set 
\[
\mydelta=\E(V)\setminus \bigcup_{x\in \X}\E(\mathcal{N}(x)),
\]
which is the set of pairs of nodes that do not have a common neighbor in $X$. The following easy lemma shows how the set $\mydelta$ is related to finding triangles which are not $\varepsilon$-heavy.
\begin{lemma}\label{lma:not-heavy-1}
Let $\varepsilon$ be any real number such that $0\le\varepsilon\le 1$.  Suppose that $\X$ is a set obtained by including each node of $V$ into $\X$ with probability $\frac{1}{9n^\varepsilon}$. Then, for any triangle $t\in T(G)\setminus\trih{\varepsilon}$, the three edges of $t$ are in $\mydelta$ with probability at least $2/3$.
\end{lemma}
\begin{proof}
Let us consider any triangle $t\in T(G)\setminus\trih{\varepsilon}$. For each edge $e\in t$ we have $\multri{e}< n^\varepsilon$. From the union bound, we obtain the inequality
$
\Pr_\X\left[
e\in 
\bigcup_{x\in \X}\E(\mathcal{N}(x))
\right]<
\frac{1}{9}.
$
Thus, again from the union bound, $t$ has its three edges in $\mydelta$ with probability at least $2/3$.
%\qed
\end{proof}
In view of Lemma~\ref{lma:not-heavy-1}, our strategy will be to take a set $\X$ as in the lemma and then find the triangles with three edges in $\mydelta$ using the following approach: Each node $i\in V$ sends to all its neighbors the set $\mathcal{N}(i)\cap \X$, in $O(|\X|)$ rounds. Then each node $k\in V$ constructs for each $j\in \mathcal{N}(k)$ the set containing all nodes $l\in \mathcal{N}(k)$ such that $\{j,l\}\in\mydelta$, which can be done locally. Let us call this set $\mathcal{S}(j,k)$ for now -- later it will be called $\mathcal{S}_V^X(j,k)$. Node $k$ then sends $\mathcal{S}(j,k)$ to neighbor~$j$, who can report all triangles of the form $\{j,k,l\}$ with $\{j,l\}\in \mydelta$.

This approach works, but its round complexity depends on the size of the sets $\mathcal{S}(j,k)$. The crucial point is that we can show from a combinatorial argument that with high probability (on the choice of~$X$) the average size of these sets is small. The main remaining issue is that, even if the average size is small, in general there exist pairs $(j,k)$ for which $\mathcal{S}(j,k)$ is large. We solve this issue by sending $\mathcal{S}(j,k)$ only when its size is close to the average and using a different strategy when its size exceeds the average. This latter strategy is based on a notion of ``good'' nodes. Roughly speaking, a node $j$ is good if it has only a small number of neighbors $k$ such that $\mathcal{S}(j,k)$ is large (see Definition~\ref{def} for the formal definition). The same combinatorial argument as above guarantees that most nodes are good (this is shown in Lemma \ref{lma:not-heavy-2} below).  Triangles containing at least one good node are fairly easy to deal with, and we can deal with triangle containing three bad nodes by applying the same approach recursively on the subgraph induced by the bad nodes. Our final algorithm will thus be recursive: it will keep a set $U\subseteq V$, with initially $U=V$, and successively search for triangles in the subgraph of $G$ induced by $U$, removing the good nodes from $U$ at each step, until $U=\emptyset$.

%============
\paragraph{Further definitions and full description of Algorithm $\Aa_3$.}
%We will need further definitions.
For any set $U\subseteq V$ and
any nodes $j,k\in U$ such that $\{j,k\}\in E$, we define the set $\myS{j,k}\subseteq U$
as follows:
\begin{align*}
\myS{j,k}&=
%\E(N_G(w))\setminus \bigcup_{u\in \X}\E(N_G(u))\\
\Big\{l\in U\:|\:  \{j,l\}\in \mydelta \textrm { and } \{k,l\}\in E\Big\}.
\end{align*}
Note that this definition is asymmetric: we have $\myS{j,k}\neq \myS{k,j}$ in general.
Let $\m$ be any positive real number. For any node $j\in U$, we define the set
\[
\myT{j}=
\big\{
k\in U\:|\:
\{j,k\}\in E 
\textrm{ and }
\left|
\myS{j,k}
\right|
> 
\m
%\]
\big\}.
\]
Finally, we will use the following concept of good nodes. 
\begin{definition}\label{def}
Let $U$ and $\X$ be any subsets of $V$, and
$\m$ be any positive real number.
A node $j\in U$ is $\m$-good for $(U,\X)$ if
$
\big|
\myT{j}
\big|
\le \m.
$
\end{definition}

The following crucial lemma shows an upper bound on the number of nodes which are not good, when $\X$ is chosen as in Lemma~\ref{lma:not-heavy-1}. 

\begin{lemma}\label{lma:not-heavy-2}
Let $\varepsilon$ be any real number such that $0\le\varepsilon\le 1$  and $\X$ be a set chosen at random as in Lemma \ref{lma:not-heavy-1}.
Then, for any real number $\m\ge \sqrt{54n^{1+\varepsilon}\log n}$, the following statement holds with probability at least $1-1/n$:
\begin{equation} \label{eq1}
\parbox{75mm}{
For any set $U\subseteq V$ there are at most $|U|/2$ nodes in $U$  that are not $\m$-good for $(U,\X)$. }
\end{equation}

\end{lemma}

\begin{proof}
Consider any pair $\{j,l\}\in\E$ such that 
\[
\multri{\{j,l\}}\ge 27n^\varepsilon\log n.
\]
%\[
%\Big|\left\{
%v\in V\:|\:
%\{u,v\}\in E \textrm{ and } \{w,v\}\in E
%\right\}\Big|
%\ge 27n^\varepsilon\log n.
%\]
We have 
\[
\Pr_\X\left[
\{j,l\}\notin 
\bigcup_{x\in \X}\E(\mathcal{N}(x))
\right]\le
\left(1-\frac{1}{9n^\varepsilon}\right)^{27n^\varepsilon\log n}
%=
%\left(1-\frac{1}{9n^\varepsilon}\right)^{9n^\varepsilon\times 3\log n}
\le
\frac{1}{n^3}.
\]
From the union bound we can thus conclude that with probability at least $1-1/n$ the following statement holds:
\begin{equation}\label{eq1b}
\multri{\{j,l\}}
<  27n^\varepsilon\log n
\textrm{
for any pair $\{j,l\}\in \mydelta$. 
}
\end{equation}

We now show that Statement (\ref{eq1b}) implies Statement (\ref{eq1}).
Let us consider the number of ordered triples $(j,k,l)\in U\times U\times U$ such that $\{j,l\}\in \mydelta$, $\{j,k\}\in E$ and $\{k,l\}\in E$. If Statement~(\ref{eq1b}) holds then this number is at most $27|U|^2n^\varepsilon \log n$.
This number of triples can also be computed in another way: it is equal to
\[
\sum_{\begin{subarray}{c}(j,k)\in U\times U\\\textrm{s.t. } \{j,k\}\in E\end{subarray}}\left|\myS{j,k}\right|.
\]
Thus if Statement~(\ref{eq1b}) holds we obtain the inequality
\[
27|U|^2n^\varepsilon \log n\ge \sum_{\begin{subarray}{c}(j,k)\in U\times U\\\textrm{s.t. } \{j,k\}\in E\end{subarray}}\left|\myS{j,k}\right| 
\ge \sum_{j\in U} \left|\myT{j}\right|\m,
\]
which implies that the number of nodes that are not $\m$-good for $(U,\X)$ is at most
\[
27|U|^2n^\varepsilon \log n\times \frac{1}{\m^2}\le  \frac{27|U|n^{1+\varepsilon} \log n}{\m^2}.
\]
For any $\m\ge \sqrt{54n^{1+\varepsilon}\log n}$ we obtain the claimed upper bound.
%\qed
\end{proof}

We now present an algorithm that lists all triangles of $G$ with three edges in $\mydelta$. The algorithm is denoted $\Aa(\X,\m)$ and described in Figure \ref{fig:algorithm}. We analyze it in the next proposition.

\begin{figure*}[ht]
\begin{center}
\fbox{
\begin{minipage}{15 cm} 
\begin{itemize}
\item[1.]
Each node tells all its neighbors whether it is in $\X$ or not.
\item[2.]
Each node $k\in V$ sends the set $\mathcal{N}(k)\cap \X$ to all its neighbors.
\item[3.]
$U\gets V$. 
\item[4.]
While $U\neq\emptyset$ do:
\begin{itemize}
\item[4.1]
Each node $k\in U$ sends the set $\myS{j,k}$ to each neighbor $j\in \mathcal{N}(k)\cap U$ for which $|\myS{j,k}|\le \m$. Each such neighbor $j$ then lists all triangles containing $j,k$ and a third node in $\myS{j,k}$, by locally computing the set $\myS{j,k}\cap \mathcal{N}(j)$.
\item[4.2]
Each node $k\in U$ decides whether itself is $\m$-good for $(U,\X)$ or not.
Let $U'\subseteq U$ denote the set of nodes that are $\m$-good for $(U,\X)$. 
\item[4.3]
Each node $j\in U'$ of the network sends the set $\myT{j}$ to each neighbor $l\in\mathcal{N}(j)\cap U$. Each such neighbor $l$ then lists all triangles containing $j$, $l$ and a third node in $\myT{j}$, by locally computing the set $\myT{j}\cap\mathcal{N}(l)$.
\item[4.4]
$U\gets U\setminus U'$. 
\item[4.5]
Each node tells all its neighbors whether it is in $U$ or not.
\end{itemize}
\end{itemize}
\end{minipage}
}
\end{center}
\caption{Algorithm $\Aa(\X,\m)$ that lists all triangles in $G$ with three edges in $\mydelta$. Here $\X$ and~$\m$ are parameters of the algorithm: $\X$ is a subset of nodes and $\m$ is a positive real number. The set~$X$ is given as follows: each node $i\in V$ knows whether $i\in \X$ or not. The number $\m$ is given to each node.}\label{fig:algorithm}
\end{figure*}

\begin{proposition}\label{prop:tri-delta}
Assume that the set $\X$ satisfies Statement (\ref{eq1}).
Then Algorithm $\Aa(\X,\m)$ terminates after $O(\log n)$ iterations of the while loop, has overall round complexity $O(|\X|+\m\log n)$, and its output $T\subseteq T(G)$ includes all the triangles of $G$ with three edges in $\mydelta$. 
\end{proposition}
\begin{proof}
Let us first show that Algorithm $\Aa(\X,\m)$ terminates after $O(\log n)$ iterations of the while loop. The assumption on the set $\X$ implies that
$
|U\setminus U'|\le |U|/2
$
at Step 4.2 of each iteration of the while loop. 
Since we start with $|U| = n$, after at most $\log_2(n)+1$ iterations of the while loop the set $U$ becomes empty, at which point we exit the while loop. 

We now analyze the overall round complexity of the algorithm by considering the complexity of each step.  Step~1 requires only one round. Note that after Step~1 each node $k$ can locally compute the set  $\mathcal{N}(k)\cap \X$.
Step~2 can then be implemented using at most $|\X|$ rounds, since $|\mathcal{N}(k)\cap \X|\le |\X|$. Step~4.1 can be implemented using at most $\m$ rounds. %Indeed, once $\myS{u,v}$ is obtained, node $u$ can locally list the triangles containing $u,v$ and a third node in $\myS{u,v}$ by computing $\myS{u,v}\cap \mathcal{N}(u)$. 
Step~4.2 does not require any communication since the enough information has been obtained at~Step 4.1 to compute $\myT{k}$ locally. Step~4.3 only requires at most $\m$ rounds, since $|\myT{j}|\le \m$ for any $j\in U'$. 
%Indeed, once this set is obtained, node $u$ can locally list all triangles containing $u$, $v$ and a third node in $\myT{v}$ by computing $\myT{v}\cap\mathcal{N}(u)\cap\mathcal{N}(v)$. 
Finally, Step~4.5 obviously requires only one round. The overall round complexity is thus at most
\[
|\X|+1+(\log_2(n)+1)(\m+\m+1)= O(|\X|+\m\log n).
\] 

We finally show that the algorithm lists all triangles of $G$ with three edges in $\mydelta$. 
First observe that for any set $U\subseteq V$, any triangle of $G$ with three nodes in $U$ and three edges in $\mydelta$ satisfies at least one of the following three properties (where $U'\subseteq U$ denotes the set of nodes that are $\m$-good for $(U,\X)$):
\begin{itemize}
\item[(a)]
it contains three distinct nodes $j,k,l\in U$ such that
$\left|\myS{j,k}
\right|
\le 
\m$ and $l\in \myS{j,k}$;
\item[(b)]
it contains two distinct nodes $j,k\in U$ such that $j\in U'$ and $k \in\myT{j}$;\item[(c)]
all its three nodes are in $U\setminus U'$.
\end{itemize}
Indeed, if such a triangle does not satisfy Property~(a) then the inequality $|\myS{j,k}|>\m$ holds for any two nodes~$j$ and~$k$ of the triangle (which means that $k\in\myT{j}$). In that case it should satisfy either Property~(b) or Property~(c).

Algorithm $\Aa(\X,\m)$ starts with $U=V$ and decreases the set $U$ until it is empty. Let us consider what happens at each execution of the while loop, with the current set $U$.
All triangles of type~(a) are found at Step 4.1 and all triangles of type~(b) are found at Step~4.3. The only remaining triangles are triangles of type~(c). They are considered at the next execution of the while loop, where $U$ is replaced by $U\setminus U'$.  
%\qed
\end{proof}

We are now ready to present Algorithm $\Aa_3$ and give the proof of Proposition \ref{prop:tri-light}.
\begin{proof}[Proof of Proposition \ref{prop:tri-light}]
Algorithm $\Aa_3$ is as follows. First each node selects itself with probability $\frac{1}{9n^{\varepsilon}}$. Then the nodes run Algorithm $\Aa(\X,\m)$ with $X$ being the set of selected nodes (note that each node $i$ of the network knows whether $i\in \X$ or not, as required) and $\m=\sqrt{54n^{1+\varepsilon}\log n}$, but stop as soon as the round complexity exceeds $c\times (n^{1-\varepsilon}+n^{(1+\varepsilon)/2}\log n)$ for some large enough constant~$c$. 

By construction, the round complexity of  Algorithm $\Aa_3$ is 
\[O(n^{1-\varepsilon}+n^{(1+\varepsilon)/2}\log n).\] 

Let $t$ be any triangle in $T(G)\setminus\trih{\varepsilon}$, and let us show that Algorithm $\Aa_3$ finds $t$ with constant probability.
Note that the expectation of the value $|\X|$ is $\frac{1}{9}n^{1-\varepsilon}$. By Chernoff bound, we know that the probability that $|\X|> \frac{2}{9}n^{1-\varepsilon}$ is negligible.
%\[
%\Pr\left[|\X|\ge 2n^{1-\varepsilon}/9\right]\le e^{-2n^{1-\varepsilon}/27}\le \frac{1}{n}.
%\]
Combining this observation with Lemmas~\ref{lma:not-heavy-1} and~\ref{lma:not-heavy-2}, we conclude (using the union bound) that with probability $\Omega(1)$ the following three conditions hold: $\X$ has size at most $\frac{2}{9}n^{1-\varepsilon}$, $\X$ satisfies Statement~(\ref{eq1}), and $t$ has its three edges in $\mydelta$. From Proposition~\ref{prop:tri-delta} we know that in this case algorithm $\Aa(X,\m)$ stops within $O(n^{1-\varepsilon}+n^{(1+\varepsilon)/2}\log n)$ rounds and finds~$t$. By choosing the constant $c$ large enough we can thus guarantee that the output of Algorithm $\Aa_3$ contains the triangle $t$ with probability $\Omega(1)$. 
%\qed
\end{proof}

%==============================================
%==============================================
\section{Lower Bounds}\label{sec:LB}
%==============================================
%==============================================
This section proves the following lower-bound theorem.

\begin{theorem} \label{thm:lowerBound}
Let $\mathcal{A}$ be any triangle listing algorithm with error probability less than $1/32$. Then there exists a probability distribution on inputs such that the expected round complexity of $\mathcal{A}$ is~$\Omega(n^{1/3}/\log n)$. 
\end{theorem}

In this section we will write $V=\{0,1,\ldots,n-1\}$. 
Without loss of generality, the run of any algorithm $\mathcal{A}$ for triangle listing on any given instance $G = (V, E)$ can be described by two following steps:
\begin{enumerate}
\item 
Each node $i$ locally constructs its initial state $\vecb{\rho}_i$ according to $G$. This state depends on the set of edges incident to $i$, but is independent of all other edges.
\item 
Each node $i$ construct its output $\vecb{T}_i$ from $\vecb{\rho}_i$ and the transcript $\vecb{\pi}_i$ of the communication received it receives during the execution of the algorithm.
\end{enumerate}

To prove Theorem~\ref{thm:lowerBound}, we consider as input the random graph $G(n, 1/2)$, which is the graph on~$n$ nodes where each pair of nodes is independently taken as an edge with probability $1/2$. Since $\vecb{\rho}_i$, $\vecb{\pi}_i$, and $\vecb{T}_i$ depend only on the algorithm $\mathcal{A}$ (and its random bits if $\mathcal{A}$ is randomized) and the input $G$, we see all of them as random variables. For any pair $\{j,k\}\in \E$, let $\vecb{e}_{\{j,k\}}$ be the random variable with value 1 if $\{j,k\}$ is an edge, and value 0 otherwise. Let $\vecb{E}=(\vecb{e}_{\{0,1\}}, \vecb{e}_{\{0,2\}}, \dots, \vecb{e}_{\{n-2,n-1\}})$ be the concatenation of these $|\E|$ random variables. The proof idea is to bound the amount of information in $\vecb{\pi}_i$ necessary to construct $\vecb{T}_i$ locally. We write $\vecb{T}=(\vecb{T}_0,\vecb{T}_1,\ldots,\vecb{T}_{n-1})$ and use the notation $w(\vecb{T})$ to represent the node identifier $i$ such that $|T_i|$ is maximum (i.e., the index of the node that outputs the maximum number of triangles).

Our proof relies on the following graph-theoretic lemma, which is shown 
in~\cite{Rivin02}.

\begin{lemma} \label{lma:triangleEdge}
If a graph $G$ contains $t$ triangles, then $G$ has at least 
$\frac{\sqrt{2}}{3}t^{2/3}$ edges.
\end{lemma}

We first recall the definition of the entropy. For a random variable $\vecb{X}$ with domain $\mathcal{X}$ and probability distribution $p$, its entropy $H(\vecb{X})$ is defined as $H(\vecb{X}) = - \sum_{x \in \mathcal{X}} p(x)\log p(x)$. The conditional entropy can be defined similarly. For any random variables $\vecb{X}$ and $\vecb{Y}$ the mutual information of $\vecb{X}$ and $\vecb{Y}$ is defined as $I(\vecb{X}; \vecb{Y}) = H(\vecb{X})-H(\vecb{X}|\vecb{Y})$. We will use in our proofs the following standard results from information theory about the mutual information.
\begin{fact} \label{fact:mutualInfo} 
For any random variables $\vecb{X}$, $\vecb{Y}$, and $\vecb{Z}$, the following three 
properties hold:
\begin{itemize}
\item $I(\vecb{X}; \vecb{Y}) = I(\vecb{Y}; \vecb{X})$,
\item $I(\vecb{X}; \vecb{Y}) \leq H(\vecb{X})$ (and $I(\vecb{X}; \vecb{Y}) \leq H(\vecb{Y})$), and
\item $I(\vecb{X}; \vecb{Y}) \geq I(\vecb{X}; (\vecb{Y}, \vecb{Z})) - 
I( \vecb{X}; \vecb{Z})$.
\end{itemize}
\end{fact}

\begin{fact}[Data Processing Inequality] \label{fact:DataProcessing}
For any three random variables $\vecb{X}, \vecb{Y}, \vecb{Z}$ such that $\vecb{X}$ 
and $\vecb{Z}$ are conditionally independent given $\vecb{Y}$, $I(\vecb{X}; \vecb{Y}) 
\geq I(\vecb{X}; \vecb{Z})$.
\end{fact}

The key technical ingredient of our proof is the following lemma 
(the notation $\Pp$ is defined in Section \ref{sec:prelim}).
\begin{lemma} \label{lma:infoLowerBoundGeneral}
For any node $i$, the inequality $I(\vecb{E}; \vecb{T}_i) \geq E[|\Pp(\vecb{T}_i)|]$ holds.
\end{lemma}

The proof of Lemma \ref{lma:infoLowerBoundGeneral} is based on
an information-theoretic argument. The intuition is very simple:  since the information $\vecb{T}_i$ completely reveals the information on the edges forming triangles in $\vecb{T}_i$, it must contain all the information on $\Pp(\vecb{T}_i)$.

We will write $M = |\E| = \frac{n(n - 1)}{2}$ and $N = |\mathcal{T}| = \frac{n(n-1)(n-2)}{6}$.
\begin{proof}[Proof of Lemma~\ref{lma:infoLowerBoundGeneral}]
%Let $M = |\E| = n(n + 1)/2$ and $N = |\mathcal{T}| = n(n+1)(n+2)/6$ for short.
Since the random variables $\vecb{e}_{\{j,k\}}$ are independent of each other, we have
\begin{align*}
\lefteqn{I(\vecb{E}; \vecb{T}_i)} \\
&= H(\vecb{E}) - H(\vecb{E}|\vecb{T}_i) \\
&= H(\vecb{E}) - \sum_{\{j,k\} \in \mathcal{E}} H(\vecb{e}_{\{j,k\}} | \vecb{T}_i) \\ 
&= H(\vecb{E}) - \sum_{\{j,k\} \in \mathcal{E}} \sum_{R \subseteq \mathcal{T}}
H(\vecb{e}_{\{j,k\}} | \vecb{T}_i = R) \cdot \Pr[\vecb{T}_i = R] \\
&= M - \sum_{R \subseteq \mathcal{T}} \sum_{\{j,k\} \in \mathcal{E}} 
H(\vecb{e}_{\{j,k\}} | \vecb{T}_i = R) \cdot \Pr[\vecb{T}_i = R]\\
&= M - \sum_{R \subseteq \mathcal{T}}\:\: \sum_{\{j,k\} \in\E\setminus\Pp(R)} 
H(\vecb{e}_{\{j,k\}} | \vecb{T}_i = R) \cdot \Pr[\vecb{T}_i = R],
\end{align*}
where the last equality comes from the assumption that Algorithm $\mathcal{A}$ never outputs a triple that is not a triangle (thus $\Pr[\vecb{T}_i = R]\neq 0$ only if $R\subseteq T(G)$, and conditioned to this event we have $\vecb{e}_{\{j,k\}}=1$ with probability $1$ for any $\{j,k\}\in\Pp(R)$). Since $H(\vecb{e}_{\{j,k\}} | \vecb{T}_i = R) \leq 1$ always holds, we conclude that
\begin{align*}
I(\vecb{E}; \vecb{T}_i) &\geq M - \sum_{R \subseteq \mathcal{T}} (M - |\Pp(R)|) \cdot \Pr[\vecb{T}_i = R] \\
&\geq \sum_{R \subseteq \mathcal{T}} |\Pp(R)| \cdot \Pr[\vecb{T}_i = R]  = E[|\Pp(\vecb{T}_i)|],
\end{align*}
as claimed.
\end{proof}

Theorem \ref{thm:lowerBound} is proved by combining Lemma~\ref{lma:triangleEdge} and Lemma~\ref{lma:infoLowerBoundGeneral}. The intuition is again fairly simple: In expectation, an instance of $G(n, 1/2)$ contains $\Omega(n^3)$ triangles, and thus the average number of triangles output per node node is $\Omega(n^2)$. Thus node $w(\vecb{T})$ outputs at least $\Omega(n^2)$ triangles and, from Lemma~\ref{lma:triangleEdge}, we obtain $E[|\Pp(\DD)|] = \Omega(n^{4/3})$. Lemma \ref{lma:infoLowerBoundGeneral} then gives the same lower bound for the mutual information between $\vecb{E}$ and $\DD$, which then gives the lower bound $\Omega(n^{4/3})$ on the amount of communication received by node $w(\vecb{T})$ during the execution of Algorithm $\Aa$. This implies the claimed $\Omega(n^{1/3}/\log n)$-round lower bound for triangle listing since $w(\vecb{T})$ can receive at most $O(n \log n)$ bits of information per round.

\begin{proof}[Proof of Theorem~\ref{thm:lowerBound}]
From Lemma~\ref{lma:infoLowerBoundGeneral}, we have 
\[
I(\vecb{E}; \DD) \geq E[|P(\DD)|].
\] 
By Lemma~\ref{lma:triangleEdge}, for any $R\subseteq T(G)$ the inequality $|\Pp(R)| \geq \frac{\sqrt{2}}{3}|R|^{2/3}$ holds. We thus have
\begin{align*}
\lefteqn{I(\vecb{E}; \DD) \geq E[|\Pp(\DD)|]} \hspace{5mm} \\
&= \sum_{R \subseteq \mathcal{T}} |P(R)| \cdot \Pr[\DD = R] \\
&\geq \sum_{R \subseteq \mathcal{T}, |R| \geq \frac{N}{16n}} 
|\Pp(R)| \cdot \Pr[\DD = R]. \\
&\geq \sum_{R \subseteq \mathcal{T}, |R| \geq \frac{N}{16n}} 
\frac{\sqrt{2}}{3}\left(\frac{N}{16n}\right)^{2/3} \cdot \Pr[\DD = R]\\
&\geq \frac{\sqrt{2}}{3}\left(\frac{N}{16n}\right)^{2/3} \cdot \Pr\left[|\DD| \geq \frac{N}{16n} \right].
\end{align*}
The expected number of triangles in $G(n,1/2)$ is $N/8$ and thus, since there cannot be more than~$N$ triangles, with probability at least $1/15$ the number of triangles exceeds $N/16$.
%Letting $\vecb{T}(G)$ be the number of triangles in $G$, it is easy to check 
%$E[\vecb{T}(G)] = N/8$. Since $\Pr[\vecb{T}(G) \geq N/16] \cdot N + 
%(1 - \Pr[\vecb{T}(G) \geq N/16] \cdot N/16 \geq E[\vecb{T}(G)] = N/8$ holds, 
%we have $\Pr[\vecb{T}(G) \geq N/16] \geq 1/ 15$. 
We conclude that the inequality $|\DD| \geq N/(16n)$ necessarily holds if the number of triangles is at least $N/16$ and Algorithm $\mathcal{A}$ correctly lists all the triangles of the graph, which occurs with overall probability probability at least $1/15 - 1/32$. Consequently
we obtain
\begin{equation}\label{eq0}
I(\vecb{E}; \DD)\geq \frac{\sqrt{2}}{3}\left(\frac{N}{16n}\right)^{2/3} \cdot \left(\frac{1}{15}-\frac{1}{32}\right) =
\Omega(n^{4/3}).
\end{equation}

From Fact~\ref{fact:mutualInfo} we have 
\begin{align}
H(\vecb{\pi}_{w(\vecb{T})}) 
&\geq I(\vecb{E}; \vecb{\pi}_{w(\vecb{T})})\nonumber \\
&\geq 
I(\vecb{E}; (\vecb{\pi}_{w(\vecb{T})}, \vecb{\rho}_{w(\vecb{T})})) 
- I(\vecb{E}; \vecb{\rho}_{w(\vecb{T})}) \nonumber \\
&\geq I(\vecb{E}; (\vecb{\pi}_{w(\vecb{T})}, \vecb{\rho}_{w(\vecb{T})})) 
- H(\vecb{\rho}_{w(\vecb{T})}).\label{eq1c}
\end{align}
Since the initial knowledge of each node $i$ is only the set of edges incident to itself, the inequality \begin{equation}\label{eq2}
H(\vecb{\rho}_{i})\le \sum_{j\neq i} H(\vecb{e}_{i,j}) = n-1
\end{equation}
holds for each node $i\in V$.
Since the output $\vecb{T}_{w(\vecb{T})}$ is computed locally only from the transcript~$\vecb{\pi}_{w(\vecb{T})}$ and the initial state $\vecb{\rho}_{w(\vecb{T})}$, the random variables $\vecb{E}$ and $\vecb{T}_{w(\vecb{T})}$ are conditionally independent given $(\vecb{\rho}_{w(\vecb{T})}, \vecb{\pi}_{w(\vecb{T})})$. We can thus use Fact \ref{fact:DataProcessing} in Inequality (\ref{eq1c}), which combined with Inequalities (\ref{eq0}) and (\ref{eq2}) gives   
\begin{align*}
H(\vecb{\pi}_{w(\vecb{T})}) 
&\geq I(\vecb{E}; \vecb{T}_{w(\vecb{T})}) - H(\vecb{\rho}_{w(\vecb{T})}) = \Omega(n^{4/3}).
\end{align*}
Since $H(\vecb{\pi}_i)$ lower bounds the average length of the transcript $\vecb{\pi}_i$, this implies that node ${w(\vecb{T})}$ receives messages of average total length $\Omega(n^{4/3})$ bits. Since node ${w(\vecb{T})}$ can receive only $O(n\log n)$ bits per round, we conclude that the expected round complexity of Algorithm $\mathcal{A}$ is $\Omega(n^{1/3} / \log n)$.
\end{proof}

The information-theoretic arguments of Lemma~\ref{lma:infoLowerBoundGeneral} can also be used to derive the following stronger lower bound mentioned in the introduction for local triangle listing, the setting where each node $i$ is required to output all the triangles containing $i$.
\begin{proposition}\label{cor}
Let $\mathcal{A}$ be any local triangle listing algorithm with error probability less than $1/32$. Then there exists a probability distribution on inputs such that the expected round complexity of $\mathcal{A}$ is $\Omega(n/\log n)$. 
\end{proposition}
The proof of Proposition \ref{cor} is similar to the proof of Theorem \ref{thm:lowerBound} but does not even requires Lemma~\ref{lma:triangleEdge}: we can immediately obtain the stronger bound $E[|\Pp(\vecb{T}_i)|] = \Omega(n^{2})$ since for local triangle listing $\vecb{T}_i$ should include all the triangles containing $i$.

\begin{proof}[Proof of Proposition~\ref{cor}]
The proof is almost similar to the proof of Theorem~\ref{thm:lowerBound}, except that we can derive the stronger lower bound
\begin{equation}\label{eq4}
I(\vecb{E}; \vecb{T}_i) =
\Omega(n^{2})
\hspace{2mm}
\textrm{ for any $i\in V$},
\end{equation}
instead of the lower bound (\ref{eq0}). This new lower bound implies that $H(\vecb{\pi}_{w(\vecb{T})})=\Omega(n^2)$ and gives the claimed lower bound on the round complexity of the local triangle listing algorithm $\Aa$.

We now explain how to derive the lower bound (\ref{eq4}). From Lemma \ref{lma:infoLowerBoundGeneral} we get
\begin{align*}
\lefteqn{I(\vecb{E};\vecb{T}_i) \geq E[|\Pp( \vecb{T}_i)|]} \\
 &= \sum_{R \subseteq \mathcal{T}} |\Pp(R)| \cdot \Pr[ \vecb{T}_i = R] \\
&\geq \frac{M}{16} 
\cdot \Pr\big[ \vecb{T}_i \textrm{ includes at least $M/6$ triangles containing $i$} \big], 
\end{align*}
since the quantity $|\Pp(R)|$ is lower bounded by the number of triangles in $R$ containing node $i$. The expected number of triangles in $G(n,1/2)$ containing any fixed node is $M/8$ and thus, since there cannot be more than~$M$ such triangles, with probability at least $1/15$ the number of triangles containing node $i$ exceeds $M/16$.  Since in Algorithm $\mathcal{A}$ node $i$ correctly lists all the triangles containing~$i$ with probability at least $1-1/32$, we conclude that $\vecb{T}_i$ includes at least $M/6$ triangles including~$i$ with overall probability at least $1/15 - 1/32$. We thus obtain 
\begin{align*}
I(\vecb{E};\vecb{T}_i)&\geq \frac{M}{16} 
\cdot \left(\frac{1}{15}-\frac{1}{32}\right) =
\Omega(n^2),
\end{align*}
as claimed
\end{proof}

\section*{Acknowledgments}
  FLG was partially supported by MEXT KAKENHI (24106009) and JSPS KAKENHI (15H01677, 16H01705, 16H05853). TI was partially supported by JSPS KAKENHI (15H00852, 16H02878) and JST-SICORP Japan-Israel Bilateral Program "ICT for a resilient society".

%=====================================
%\bibliographystyle{plain}
%\bibliography{refs}

\begin{thebibliography}{10}

\bibitem{ABI86}
Noga Alon, Laszlo Babai, and Alon Itai.
\newblock A fast and simple parallel algorithm for the maximal independent set
  problem.
\newblock {\em Journal of Algorithms}, 7(4):567--583, 1986.

\bibitem{Bansal+ToC12}
Nikhil Bansal and Ryan Williams.
\newblock Regularity lemmas and combinatorial algorithms.
\newblock {\em Theory of Computing}, 8(1):69--94, 2012.

\bibitem{BE13}
Leonid Barenboim and Michael Elkin.
\newblock {\em Distributed Graph Coloring: Fundamentals and Recent
  Developments}.
\newblock Morgan {\&} Claypool, 2013.

\bibitem{BEPS16}
Leonid Barenboim, Michael Elkin, Seth Pettie, and Johannes Schneider.
\newblock The locality of distributed symmetry breaking.
\newblock {\em Journal of the ACM}, 63(3):20:1--20:45, 2016.

\bibitem{Censor-Hillel+DISC16}
Keren Censor{-}Hillel, Eldar Fischer, Gregory Schwartzman, and Yadu Vasudev.
\newblock Fast distributed algorithms for testing graph properties.
\newblock In {\em Proceedings of the International Symposium on Distributed
  Computing (DISC)}, pages 43--56, 2016.

\bibitem{CKKLPS15}
Keren Censor-Hillel, Petteri Kaski, Janne~H. Korhonen, Christoph Lenzen, Ami
  Paz, and Jukka Suomela.
\newblock Algebraic methods in the congested clique.
\newblock In {\em Proceedings of the ACM Symposium on Principles of Distributed
  Computing (PODC)}, pages 143--152, 2015.

\bibitem{Czumaj+SICOMP09}
Artur Czumaj and Andrzej Lingas.
\newblock Finding a heaviest vertex-weighted triangle is not harder than matrix
  multiplication.
\newblock {\em SIAM Journal on Computing}, 39(2):431--444, 2009.

\bibitem{DLP12}
Danny Dolev, Christoph Lenzen, and Shir Peled.
\newblock ``{Tri}, tri again": Finding triangles and small subgraphs in a
  distributed setting.
\newblock In {\em Proceedings of the International Symposium on Distributed
  Computing (DISC)}, pages 195--209, 2012.

\bibitem{DKO14}
Andrew Drucker, Fabian Kuhn, and Rotem Oshman.
\newblock On the power of the congested clique model.
\newblock In {\em Proceedings of the ACM Symposium on Principles of Distributed
  Computing (PODC)}, pages 367--376, 2014.

\bibitem{Elkin04}
Michael Elkin.
\newblock A faster distributed protocol for constructing a minimum spanning
  tree.
\newblock In {\em Proceedings of the ACM-SIAM Symposium on Discrete Algorithms
  (SODA)}, pages 359--368, 2004.

\bibitem{Elkin06}
Michael Elkin.
\newblock An unconditional lower bound on the time-approximation trade-off for
  the distributed minimum spanning tree problem.
\newblock {\em SIAM Journal on Computing}, 36(2):433--456, 2006.

\bibitem{Fraigniaud+DISC16}
Pierre Fraigniaud, Ivan Rapaport, Ville Salo, and Ioan Todinca.
\newblock Distributed testing of excluded subgraphs.
\newblock In {\em Proceedings of the International Symposium on Distributed
  Computing (DISC)}, pages 342--356, 2016.

\bibitem{FHW12}
Silvio Frischknecht, Stephan Holzer, and Roger Wattenhofer.
\newblock Networks cannot compute their diameter in sublinear time.
\newblock In {\em Proceedings of the ACM-SIAM Symposium on Discrete Algorithms
  (SODA)}, pages 1150--1162, 2012.

\bibitem{GHS83}
Robert~G. Gallager, Pierre~A. Humblet, and Philip~M. Spira.
\newblock A distributed algorithm for minimum-weight spanning trees.
\newblock {\em ACM Transactions on Programming Languages and Systems (TOPLAS)},
  5(1):66--77, 1983.

\bibitem{GKP93}
J.A. Garay, S.~Kutten, and D.~Peleg.
\newblock A sub-linear time distributed algorithm for minimum-weight spanning
  trees.
\newblock In {\em Proceedings of the IEEE Symposium on Foundations of Computer
  Science (FOCS)}, pages 659--668, 1993.

\bibitem{GKKLP15}
Mohsen Ghaffari, Andreas Karrenbauer, Fabian Kuhn, Christoph Lenzen, and Boaz
  Patt-Shamir.
\newblock Near-optimal distributed maximum flow: Extended abstract.
\newblock In {\em Proceedings of the ACM Symposium on Principles of Distributed
  Computing (PODC)}, pages 81--90, 2015.

\bibitem{GK13}
Mohsen Ghaffari and Fabian Kuhn.
\newblock Distributed minimum cut approximation.
\newblock In {\em Proceedings of the International Symposium on Distributed
  Computing (DISC)}, pages 1--15, 2013.

\bibitem{Hirvonen+14}
Juho Hirvonen, Joel Rybicki, Stefan Schmid, and Jukka Suomela.
\newblock Large cuts with local algorithms on triangle-free graphs.
\newblock {\em CoRR}, abs/1402.2543, 2014.

\bibitem{HW12}
Stephan Holzer and Roger Wattenhofer.
\newblock Optimal distributed all pairs shortest paths and applications.
\newblock In {\em Proceedings of the ACM Symposium on Principles of Distributed
  Computing (PODC)}, pages 355--364, 2012.

\bibitem{Itai+SICOMP78}
Alon Itai and Michael Rodeh.
\newblock Finding a minimum circuit in a graph.
\newblock {\em SIAM Journal on Computing}, 7(4):413--423, 1978.

\bibitem{KMW16}
Fabian Kuhn, Thomas Moscibroda, and Roger Wattenhofer.
\newblock Local computation: Lower and upper bounds.
\newblock {\em Journal of the ACM}, 63(2):17:1--17:44, 2016.

\bibitem{FR05}
Fabian Kuhn and Roger Wattenhofer.
\newblock Constant-time distributed dominating set approximation.
\newblock {\em Distributed Computing}, 17(4):2005, 303--310.

\bibitem{LP15}
Christoph Lenzen and Boaz Patt-Shamir.
\newblock Fast partial distance estimation and applications.
\newblock In {\em Proceedings of the ACM Symposium on Principles of Distributed
  Computing (PODC)}, pages 153--162, 2015.

\bibitem{LP13}
Christoph Lenzen and David Peleg.
\newblock Efficient distributed source detection with limited bandwidth.
\newblock In {\em Proceedings of the ACM Symposium on Principles of Distributed
  Computing (PODC)}, pages 375--382, 2013.

\bibitem{Linial92}
Nathan Linial.
\newblock Locality in distrubuted graph algorithms.
\newblock {\em SIAM Journal on Computing}, 21(1):193--201, 1992.

\bibitem{LPPP05}
Zvi Lotker, Boaz Patt-Shamir, Elan Pavlov, and David Peleg.
\newblock Minimum-weight spanning tree construction in {$O(\log\log n)$}
  communication rounds.
\newblock {\em SIAM Journal on Computing}, 35(1):120--131, 2005.

\bibitem{Nanongkai14}
Danupon Nanongkai.
\newblock Distributed approximation algorithms for weighted shortest paths.
\newblock In {\em Proceedings of the ACM Symposium on Theory of Computing
  (STOC)}, pages 565--573, 2014.

\bibitem{NS14}
Danupon Nanongkai and Hsin-Hao Su.
\newblock Almost-tight distributed minimum cut algorithms.
\newblock In {\em Proceedings of the International Symposium on Distributed
  Computing (DISC)}, pages 439--453, 2014.

\bibitem{Pandurangan+16}
Gopal Pandurangan, Peter Robinson, and Michele Scquizzato.
\newblock Tight bounds for distributed graph computations.
\newblock {\em CoRR}, abs/1602.08481, 2016.

\bibitem{PatrascuSTOC10}
Mihai Patrascu.
\newblock Towards polynomial lower bounds for dynamic problems.
\newblock In {\em Proceedings of the ACM Symposium on Theory of Computing
  (STOC)}, pages 603--610, 2010.

\bibitem{Peleg00}
David Peleg.
\newblock {\em Distributed computing: a locality-sensitive approach}.
\newblock Society for Industrial and Applied Mathematics, 2000.

\bibitem{PR99}
David Peleg and Vitaly Rubinovich.
\newblock A near-tight lower bound on the time complexity of distributed
  minimum-weight spanning tree construction.
\newblock {\em SIAM Journal on Computing}, 30(5):1427--1442, 2000.

\bibitem{Pettie+15}
Seth Pettie and Hsin{-}Hao Su.
\newblock Distributed coloring algorithms for triangle-free graphs.
\newblock {\em Information and Computation}, 243:263--280, 2015.

\bibitem{Rivin02}
Igor Rivin.
\newblock Counting cycles and finite dimensional {$L^p$} norms.
\newblock {\em Advances in Applied Mathematics}, 29(4):647 -- 662, 2002.

\bibitem{SHKKANPPW12}
Atish~Das Sarma, Stephan Holzer, Liah Kor, Amos Korman, Danupon Nanongkai,
  Gopal Pandurangan, David Peleg, and Roger Wattenhofer.
\newblock Distributed verification and hardness of distributed approximation.
\newblock {\em SIAM Journal on Computing}, 41(5):1235--1265, 2012.

\bibitem{Williams+FOCS10}
Virginia {Vassilevska Williams} and Ryan Williams.
\newblock Subcubic equivalences between path, matrix and triangle problems.
\newblock In {\em Proceedings of the IEEE Symposium on Foundations of Computer
  Science (FOCS)}, pages 645--654, 2010.

\bibitem{Vassilevska+SICOMP13}
Virginia {Vassilevska Williams} and Ryan Williams.
\newblock Finding, minimizing, and counting weighted subgraphs.
\newblock {\em SIAM Journal on Computing}, 42(3):831--854, 2013.

\bibitem{WC81}
Mark~N. Wegman and J.~Lawrence Carter.
\newblock New hash functions and their use in authentication and set equality.
\newblock {\em Journal of Computer and System Sciences}, 22(3):265 -- 279,
  1981.

\end{thebibliography}
%=====================================

\end{document}